\documentclass[10pt,journal,compsoc]{IEEEtran}
\ifCLASSOPTIONcompsoc
  \usepackage[nocompress]{cite}
\else
  \usepackage{cite}
\fi

\usepackage{cite}
\usepackage{amsmath,amssymb,amsfonts}
\usepackage{algorithmic}
\usepackage{graphicx}
\usepackage{textcomp}
\usepackage{xcolor}
\usepackage{amsmath,amsthm}
\usepackage{multirow}
\usepackage{breqn}

\def\BibTeX{{\rm B\kern-.05em{\sc i\kern-.025em b}\kern-.08em
    T\kern-.1667em\lower.7ex\hbox{E}\kern-.125emX}}

\DeclareMathOperator{\EX}{\mathbb{E}}
\DeclareSymbolFont{bbold}{U}{bbold}{m}{n}
\DeclareSymbolFontAlphabet{\mathbbold}{bbold}

\newtheorem{theorem}{\bf Theorem}

\newtheorem{proposition}{\bf Proposition}
\newtheorem{lemma}{\bf Lemma}

\makeatletter
\def\maketag@@@#1{\hbox{\m@th\normalfont\normalsize#1}}
\makeatother

\begin{document}

\title{\huge Finding the Sweet Spot for Data Anonymization:\\ A Mechanism Design Perspective}

\author{Abdelrahman Eldosouky, Tapadhir Das, Anuraag Kotra, and Shamik Sengupta
 \vspace{-0.5cm}
 \thanks{This paper is an extension of the work originally presented in \cite{korta2020Every}}
 }


\IEEEtitleabstractindextext{%
\begin{abstract}
Data sharing between different organizations is an essential process in today’s connected world. However, recently there were many concerns about data sharing as sharing sensitive information can jeopardize users’ privacy. To preserve the privacy, organizations use anonymization techniques to conceal users' sensitive data. However, these techniques are vulnerable to de-anonymization attacks which aim to identify individual records within a dataset. In this paper, a two-tier mathematical framework is proposed for analyzing and mitigating the de-anonymization attacks, by studying the interactions between sharing organizations, data collector, and a prospective attacker. In the first level, a game-theoretic model is proposed to enable sharing organizations to optimally select their anonymization levels for k-anonymization under two potential attacks: background-knowledge attack and homogeneity attack. In the second level, a contract-theoretic model is proposed to enable the data collector to optimally reward the organizations for their data. The formulated problems are studied under single-time sharing and repeated sharing scenarios. Different Nash equilibria for the proposed game and the optimal solution of the contract-based problem are analytically derived for both scenarios. Simulation results show that the organizations can optimally select their anonymization levels, while the data collector can benefit from incentivizing the organizations to share their data. \end{abstract}

}

\maketitle
\IEEEdisplaynontitleabstractindextext
\IEEEpeerreviewmaketitle

\IEEEraisesectionheading{\section{Introduction}\label{sec:introduction}}

\IEEEPARstart{T}{he} rise of Big Data has helped generate tremendous amounts of digital information that are continually being collected, analyzed, and distributed. This technology has helped organizations personalize their services, optimize their decision making, and help predict future trends \cite{zyskind2015decentralizing}. Nevertheless, these operations tend to raise public concern due to the fact that much of the data contain user sensitive information. To address these concerns and preserve user privacy, organizations engage in deploying robust security mechanisms to protect their data against different forms of cyber-attacks \cite{badsha2019privacy}. Consequently, concepts like data security, privacy, and trust have recently received significant attention in the literature as different forms of preserving the data \cite{soria2017individual,zhu2014correlated, keshavarz2020real, afghah2020cooperative,boreale2019relative,Domingo2019Steered}.

Yet, conventional security mechanisms do not become handy when it comes to data sharing. For instance, encryption based mechanisms can help to secure data shared between different parts or sites of the same organization, e.g., patients' remote monitoring \cite{eldosouky2018cybersecurity}. However, it is not feasible to widely share encrypted data, among many organizations, due to key management issues. One solution for data sharing is to remove the sensitive information, e.g., name, phone number, and address from the a dataset, before sharing it. However, it was shown that the remaining unique characteristics of the dataset can still be used to identify users {\cite{sweeney2002k}}. To further preserve the privacy when sharing data, anonymization techniques have been proposed to ensure that each record in a dataset is indistinguishable from the others, by removing identifiable features, and, hence, reducing the probability of identifying individual records. Prominent anonymization methods include $k$-anonymization {\cite{sweeney2002achieving}}, $l$-diversity {\cite{machanavajjhala2006diversity}}, $t$-closeness {\cite{li2007t}}, where $k$-anonymization was introduced first and then $l$-diversity and $t$-closeness were introduced as expansions on it that provide further modifications to the dataset, making it more challenging to differentiate the rindividual ecords.

Despite these developments, specific de-anonymization techniques, like background knowledge attacks \cite{li2009modeling} and homogeneity attacks \cite{wang2011enhanced}, are able to compromise the security of these approaches, i.e., $k$-anonymization, $l$-diversity, and $t$-closeness. As, $k$-anonymization is the basic technique behind $l$-diversity, and $t$-closeness, the research conducted in this paper focuses primarily on $k$-anonymization. In $k$-anonymization, the magnitude of $k$ directly corresponds to the level of privacy achieved on the dataset. However, it also corresponds to the amount of information loss from the dataset, which may reduce its usefulness to other organizations if the value of $k$ is too big. The goal for the organizations is, then, to optimally select $k$ such that it maximizes privacy while minimizing information loss. The work in \cite{liang2008infoloss} proposed two algorithms to reduce the information loss associated with $k$-anonymization. However, as these algorithms depend on the structure of the data, they cannot be generalized. To this end, choosing the optimal value of $k$, in $k$-anonymization, remains an open problem. 

\subsection{Related Work}

Different techniques have been proposed to enable the organizations to preserve the privacy of their shared information \cite{xu2016dynamic, baza2019b,de2017pracis,baza2020sharing,zhao2012collaborative,alawneh2008preventing}. The authors in \cite{xu2016dynamic} studied the case of asymmetric information sharing in which a data collector interacts with multiple data owners sequentially and each data owner possesses a record desired by the collector. The authors proposed a pricing technique to enable the data collector to determine the optimal value of the data. The work in \cite{baza2019b} addressed the privacy issue of data sharing using blockchains, in which the authors proposed a novel and efficient protocol to mitigate Sybil attacks by malicious users through a time-locked deposit protocol while minimizing the execution costs.
The work in {\cite{de2017pracis}} introduced a data analytics system for privacy preservation in data forwarding and aggregation, by using summary statistics of encrypted value aggregations.
The authors in \cite{baza2020sharing} proposed a technique to share trained machine learning models instead of the original data in the applications that use the data for prediction purposes.
Meanwhile, the works in \cite{zhao2012collaborative} and \cite{alawneh2008preventing} investigated the concept of data sharing using an information sharing platform, with the goal of achieving collaborative information sharing between the organizations.
For instance, the work in \cite{zhao2012collaborative} proposed a collaborative information sharing environment to provide cyber incident prevention, protection for the shared data.
The work in {\cite{alawneh2008preventing}} addressed the problem of information leakage by preventing employees in collaborating organizations from transferring the sensitive information accidentally or deliberately to non-authorized users. However, one limitation of the works in \cite{xu2016dynamic, baza2019b,de2017pracis,baza2020sharing,zhao2012collaborative,alawneh2008preventing} is that they focus on preserving the privacy of the shared data from a passive standpoint, i.e. they do not consider active attacks that try to directly reveal the sensitive information of the dataset.


Other works in literature have studied the privacy, under possible attacks, by analyzing the interactions between the sharing organizations and the attacker using game theory \cite{han2012game}. Game theory is a powerful mathematical framework that enables to study the interactions between different decision makers and that is widely used in many security domains \cite{das2020think,Eldosouky2020drones,Ferdowsi2020Interdependence}. 
Similarly, game theory has been recently used to preserve the privacy of shared data \cite{vakilinia20173,wan2017expanding,ezhei2017information}. For instance, the work in {\cite{vakilinia20173}} studied the trade-off between information sharing and the cost of privacy and preservation based on incentives from the information sharing platform. The authors in \cite{wan2017expanding} developed a game-theoretic approach to share genomic data that accounts for the adversarial behavior and the available resources. The work in \cite{ezhei2017information} studied the problem of sharing security data for investment purposes and the authors have proposed a game-theoretic approach to make decisions based on the privacy risk and the security knowledge needs. However, while the works in  \cite{vakilinia20173,wan2017expanding,ezhei2017information} help to preserve the privacy of the shared data, their approaches do not apply to data anonymization and, hence, the problem of optimizing the anonymization level.


Finally, we note that game theory has been also used for data anonymization \cite{liu2013game,chakravarthy2012coalitional,Adl2012privacy}. In \cite{liu2013game}, the authors introduced a game-theoretic approach to ensure $k$-anonymity when generating dummy data for a location-based service that involves multiple users. The authors in \cite{chakravarthy2012coalitional} used a coalitional game-theoretic model to fix the anonymization level of $k$-anonymity based on a given threshold for information loss. However, the works in \cite{liu2013game} and \cite{chakravarthy2012coalitional} are based on the assumption of a fixed anonymization level and do not enable to optimize its selection. The authors in \cite{Adl2012privacy} introduced an approach to optimize the anonymization level selection in a scenario consisting of three different parties: a data provider, a data collector, and a data user. The equilibrium of the game formulated in \cite{Adl2012privacy} was analytically derived to choose a value of $k$ that represents the shared agreements between the different parties. However, the work in \cite{Adl2012privacy} does not consider the presence of an attacker which can affect the utilities of different parties.


\subsection{Contributions}
The main contribution of this paper is a general framework to optimize the selection of anonymization levels under possible attacks. In particular, we propose a multi-level framework to study the interactions between the different entities involved in data sharing, i.e., the sharing organizations, an information sharing platform (data collector), and an attacker.
In the first level, we formulate a game-theoretic model to analyze the interactions between the organizations and the attacker where each organization chooses a value of $k$ that maximizes its outcome based on the expected attacks and the choices of other organizations. Meanwhile, an attacker can choose from a set of attacks based on its expected outcome. The framework considers two types of de-anonymization attacks which are background knowledge and homogeneity attack \cite{li2009modeling} and \cite{wang2011enhanced}. First, we consider the case of single-time sharing which corresponds to a static non zero-sum game, for which we analytically derive both the pure and the mixed-strategy equilibrium points. Then, we solve a dynamic game, which represents a repeated sharing scenario, by tracing the change of the utilities over time.

In the second level of the framework, we formulate an optimal contract-theoretic problem to study the interactions between the various organizations and the the data collector. In particular, the data collector offers contracts to the organizations that maximize its own reward while incentivizing each organization to accept a contract and to share its data with the data collector. The problem is formulated using the framework of contract theory \cite{bolton2005contract} that provides a set of tools for modeling the relations between a principal (the data collector) and a number of agents (the organizations). Note that, contract theory has been used in a wide variety of applications to solve the principal-agents problem, e.g., {\cite{eldosouky2020resilient,duan2012cooperative,contract_Abdel}}. However, to the best of our knowledge, this is the first work to use contract theory in information sharing problems to mitigate de-anonymization attacks. To this end, the optimal solution of the contract based problem is analytically derived under the single-time sharing scenario. Then, we propose an approach in the repeated sharing scenario to incentivize more organizations to share their data.





The rest of the paper is organized as follows.
The system model and the proposed two-tier framework are formulated in Section{~\ref{chapter_three}}.
The equilibrium analysis of the proposed game and the optimal solution of the contract-based model are derived for the static case in Section{~\ref{chapter_four}}.
The solutions of the dynamic case are derived in Section{~\ref{chapter_five}}.
Numerical results are presented and analyzed in Section~\ref{chapter_six}. Finally, conclusions are drawn in Section~\ref{chapter_seven}.

\section{System Model}\label{chapter_three}

\subsection{Background}
The goal of $k-$anonymization is to make each record in the shared dataset indistinguishable from at least $k-1$ other records \cite{sweeney2002achieving}. This can be achieved by applying some operations on the dataset attributes such as generalization and suppression. In general, the attributes of a dataset can be classified into key attributes such as name and address, quasi- identifiers such as date of birth, zip code, and gender, and sensitive attributes that are specific attributes such as medical records and salaries. When sharing a dataset, the researches are interested in the sensitive attributes which indicate the valuable information in the dataset. To preserve the privacy of the individuals in the dataset, key attributes are always removed before sharing. On the other hand, quasi-identifiers are the attributes that are processed to achieve the k-anonymity. For instance, in the generalization process, the quasi-identifiers are replaced with less specific values, e.g., removing the last two digits from the zip code. Suppression is used to remove specific records, usually outliers, that if kept will cause too much information loss to achieve a specific $k$ value under the generalization process.

After achieving a specific $k$ value, the probability of identifying an individual record is reduced to $1/k$. However, some attacks can increase this probability for an attacker. For instance, under background knowledge attack, the attacker collects some information to help eliminate some values from the dataset, so that, it increases the probability of identifying specific records ~\cite{wang2017modeling}. On the other hand, under homogeneity attack, the sensitive attributes might be the same for more than one record among the $k$ indistinguishable records, which will increase the probability to infer that sensitive information of a specific record ~\cite{machanavajjhala2006diversity}. Note that, in this work, we do not consider special datasets with frequent common values in the sensitive attributes that can self-reveal the sensitive attributes even after being anonymized. 

\subsection{System Model}
Consider an information exchange scenario in which some organizations interact with a data collector to share sets of data. The data collector is a system that collects data from different organizations and manages how this data is shared later for different purposes, e.g., research, marketing, etc. Typically, the shared data can contain sensitive information, therefore, the organizations usually apply an anonymization technique to their data before sharing it. We also consider the presence of an attacker that targets the shared data, at the data collector side, by applying de-anonymization techniques and extracting the sensitive information. Fig. \ref{fig:system} shows the interactions in our system model for the case of two organizations. As discussed earlier, we consider the case where the organizations perform $k$-anonymization technique to preserve the privacy. Here, we assume that the organizations are able to achieve the desired $k$ level by using both generalization and suppression techniques. 

\begin{figure}[t]
    \centering
    \includegraphics[width=8cm]{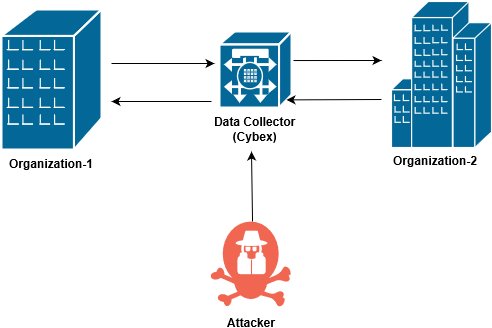}
    \caption{The interactions between the organizations, the data collector, and the attacker.}
    \label{fig:system}
\end{figure}

As there are three type of parties in the system, we propose a two-tier model to model the different interactions between the parties. First, we use game theory to model the interactions between the organizations and the attacker, to enable the organizations to strategically choose the appropriate level of data anonymization. In the second tier, we formulate a contract-theoretic problem for the data collector to optimize its rewards from the shared data by determining the optimal payments to the organizations. Next, we define the different entities in the system, following which, we formulate the mathematical models for both the system's tiers.

\subsubsection{Entities}
\textbf{Organizations}: Consider a set $\mathcal{N}$ of $N$ organizations that share their data. Each organization uses $k-$anonymization technique to make its shared data anonymous. The goal of organization $i$ is to choose the best value of $k_i$ to maximize its payoff, given other organizations $k$ values and the possible attacks on the data. Each organization can take an action, choosing an anonymization level from a set $\mathcal{D}$ of available levels.

\textbf{Attacker}: An attacker targets the data, at the data collector side, in order to reveal the sensitive information. We assume the attacker can anticipate the level of anonymization used, by analyzing the structure of the dataset. The attacker has three actions to choose from. Let $a \in \mathcal{A}=\left\{B, H, N\right\}$ represents the attacker's action of performing \emph{background knowledge} attack, performing \emph{homogeneity attack}, or no-attack, respectively.

\textbf{Data Collector}: A data collector is a system, or independent organization, whose objective is to collect data from the different organizations and earn profits by performing data mining. Here, the data collector needs to find the appropriate rewards, so that it can achieve a positive outcome while incentivizing the organizations to share their data.

\subsubsection{First Tier: A Game-theoretic Model}\label{Sec:first_Tier}
In the first tier of the framework, we study the interactions between the organizations and the attacker. Since the attacker has opposing goals to the organizations, game theory \cite{han2012game} represents a suitable mathematical framework for modeling such interactions. In the proposed game, each player wants to maximize its payoff based on its actions and the other players' actions. The players' payoffs are given in the form of utility functions which map the players' actions to their outcome as discussed next.

\textbf{Organizations}: The utility of each organization is given as a function in the reward it gets from the data collector, $r_i(k_i)$, the cost for applying the anonymization technique, $c_i(k_i)$, the level of trust in the data collector, $T_i(k_i)$, and the probability of data breach, $b_i(k_i,k_{-i},a)$, where $k_{-i}$ refers to the other organization's action. 
Let $u_i$ be the utility of organization $i$, it can then be given by:
\begin{equation}\label{eq:defenderPayOff}
u_i(k_i,k_{-i},a) = r_i(k_i) \cdot (1 - b_i(k_i,k_{-i},a)) - c_i(k_i) + T_i(k_i),
\end{equation}
where the first term represents the probability of receiving the reward based on all players' actions. The reward function $r_i(k_i)$ is given by the data collector, as discussed later.

The cost function $c_i(k_i)$ depends on the choice of $k_i$. Here, we propose to define the anonymization cost as a function in the computational complexity of executing $k$-anonymization procedure. This complexity was shown in \cite{meyerson2004complexity} to equal $\mathcal{O}(|V|^{2k})$, where $|V|$ represents the number of different subsets of the dataset based on their common attributes for anonymization purpose. Using this complexity function, we propose to define the cost function $c_i(k_i)$ as:

\begin{equation}\label{eq:cost1}
c_i(k) = log \left( |V|^{2k} \right), 
\end{equation}
where the $log$ function converts the time complexity into monetary values on the same scale as the reward values.

Next, we consider the level of trust in the data collector $T_i(k_i)$, which represents how much an organization trusts the data collector to protect its dataset against cyber threats, e.g., breaches. We propose to define the level of trust as:
\begin{equation}\label{eq:trust}
    T_i(k_i)= \gamma \cdot k_i
\end{equation}
where, $\gamma$ is the coefficient of trust. Note that, the level of trust is defined as an increasing function in $k_i$ as when the anonymization level increases, the organizations will be more confident that their information will be safe even under breaches as it is less informative.

Finally, we consider the breach probability for each organization's shared data, $b_i(k_i,k_{-i},a)$. We use the probability  of breach defined in~\cite{Gordon:2002:EIS:581271.581274} as:
\begin{equation}\label{eq:co-efficient of loss}
    b(a,k_i) = \frac{p(a)}{\alpha k_i+1},
\end{equation}
Where $p(a)$ is probability of a successful attack, based on the attack type, and $\alpha>0$ is a measure of information security. Note that, in \cite{Gordon:2002:EIS:581271.581274}, the probability of breach is given as a function in the organization's investment which is analogous to the level of anonymization $k_i$ as the organization's investment in protecting its shared data. 

Equation \eqref{eq:co-efficient of loss} represents the organization's own probability of breach. In case of multiple organizations sharing to the same platform, this will increase the probability of successful attacks as an attacker can link information from different datasets to identify the records~\cite{sattar2014probabilistic}. Here, we propose to model this interdependency similar to the model in~\cite{ogut2005cyber} such that the interdependent risk between two organizations is given as:
\begin{equation}\label{eq:dependentrisk}
    b(a,k_i,k_{-i})  = 1- (1- \frac{p(a)}{\alpha k_i+1})(1- \frac{p(a)}{\alpha k_{-i}+1}). 
\end{equation}

Substituting \eqref{eq:dependentrisk} in \eqref{eq:defenderPayOff}, the utility of any organization, in terms of its and other organizations' anonymization levels, can then be given as:
\begin{dmath}\label{eq:updatedutility}
    u_i(k_i,k_{-i},a) = r_i(k_i) \cdot (1- \frac{p(a)}{\alpha k_i+1}) \cdot (1- \frac{p(a)}{\alpha k_{-i}+1}) - c_i(k_i)+\gamma \cdot k_i.
\end{dmath}


\textbf{Attacker}: For the case of two organizations, the attacker's utility can be given in terms of its probability of achieving the reward from the information and the cost to apply its attack. Thus, we define the attacker's utility $u_a$ as follows:
\begin{equation}\label{eq:attackerPayOff}
u_a(k_1,k_2,a)= b(a,k1,k_2)R_a-c_a(a),
\end{equation}
where $R_a$ is the reward for revealing the real data and this reward can be achieved based on the combined breach probabilities of the datasets and $c_a(a)$ is the cost of performing each type of the attack. 

Note that \eqref{eq:dependentrisk} can be rewritten as:
\begin{equation}
b(a,k_1,k_2) = \frac{p(a)}{(\alpha k_1+1)} + \frac{p(a)}{(\alpha k_2+1)} - \frac{p(a)}{(\alpha k_1+1)} \frac{p(a)}{(\alpha k_2+1)},
\end{equation}
and, thus, the attacker's utility in \eqref{eq:attackerPayOff} can be written as:
\begin{dmath}\label{eq:attackerPayOff2}
u_a(k_1,k_2,a)= \Big(\frac{p(a)}{(\alpha k_1+1)} + \frac{p(a)}{(\alpha k_2+1)} - \frac{p(a)}{(\alpha k_1+1)} \frac{p(a)}{(\alpha k_2+1)}\Big)R_a-c_a(a),
\end{dmath}

Here, according to the nature of the homogeneity attack, the attacker will benefit if the two organizations are using the same anonymization level for sharing datasets with related information. This is because the similar structure of the shared data will increase the probability of identifying the individual records.

Let $p(H_s)$ be the success probability of the homogeneity attack when the organizations use the same anonymization level. Similarly, let $p(H_d)$ be the success probability of the homogeneity attack when the organizations use different anonymization levels, such that $p(H_s)>p(H_d)$. We assume $p(B)>p(H_d)>0$, i.e., the success probability of background attack is higher than that of the homogeneity attack. That is because the attacker can use the \emph{background knowledge} to link with the shared data and increase its chances of identifying the records. However, $p(B)$ can be higher or lower than $p(H_s)$. 

For the attack cost, the cost of performing the background knowledge attack is assumed to be higher than that of the homogeneity attack, i.e., $c_a(B)>c_a(H)>0$. This is because the attacker will spend more time collecting the background information and linking the similar information. Note that, when the attacker chooses not to attack, its utility $u_a(N,k_1,k_2)$ will equal zero. This choice will be superior to the attacker if the cost of performing the attack exceeds the reward from revealing the information.

To this end, we define a game $\mathcal{G}$= $\{\mathcal{N}$, $\mathcal{D}$, $\mathcal{A}$, $\mathcal{U}\}$ such that $\mathcal{N}$ is the set of the players which include all the organizations as well as the attacker, $\mathcal{D}$ is the set of the organization's strategies, $\mathcal{A}$ is the set of attacker's strategies, and $\mathcal{U}$ is the set of the all players' utilities. 
 
 \subsubsection{Second Tier: A Contract-theoretic Model}\label{Sec:second_Tier}

In the second tier of the framework, we study the interactions between the data collector and the organizations. We notice that the data collector gives rewards to the organizations for their shared data. Since the data collector has the power to price the data, it can make take-it or leave-it offers, i.e., if an organization was offered a low price, it can refuse it and does not share its data with the data collector. To this end, we propose to use contract theory \cite{bolton2005contract} to model the data collector's problem of finding the optimal prices for the data.

The goal of the data collector is to collect the data from the organizations and make profits by performing data mining on it. The data collector will be referred to as the principal in this section. The utility function of the principal can be given as:
\begin{equation}\label{cybex}
    U_d = \sum_{i=1}^{N} \theta_i \left( v_i -  r_i(k_i) \right),
\end{equation}
where, $v_i$ is the principal's evaluation of the received data from organization $i$ and $r_i(k_i)$ is the reward paid to organization $i$ for its data. The evaluation $v_i$ represents the data collector's expected profits from obtaining this data. Finally, $\theta_i$ is the organization's type, which specifies how the principal perceives different organizations in the market such that $0 \le \theta_i \le 1$. Since the main difference between the organizations is their $k$ selection, $\theta_i$ needs to be a function in $k_i$.

Here we propose to define $\theta = 1 / k$ such that it satisfies $0 \le \theta_i \le 1$. Moreover, the principal's utility will be a declining function in $k$ such that when the level of anonymization increases, the information will be less informative and, hence, its value will decrease. In return, the data collector will give less rewards to the organizations for their data.


Note that, by using $\theta = 1 / k$ when $k=1$, i.e., no anonymization, the data collector can obtain the full value of the reward and the organizations can obtain the full reward. For every $k>1$, the reward will be declining such that, for large values of $k$, e.g., $k>10$, any increase in $k$ will cause a small decrease in $r_i$. This can be interpreted as when the anonymization level increases, the information will be less useful up to some point where the increased $k$ will have very small effect on the information loss (reward). This can be captured by the heavy tail of the function $\theta = 1 / k$.


The principal's problem is to design different contracts for the different organization types such that organizations from the same type will be given equal payments. The contracts offered
by the principal need to be feasible for the organizations, i.e., they need to be persuading for the organizations to accept. To this end, the contracts must satisfy two key properties \cite{bolton2005contract} which are individual rationality (IR) and incentive compatibility (IC).

\begin{enumerate}
    \item \textbf{Individual Rationality (IR):} As the organizations are rational, the given payments need to ensure non-negative utility for each organizations, i.e.,
    \begin{equation}\label{ineq:1}
        \theta_i r_i(k_i) - c_i(k_i) + T_i(k_i) \ge 0,~~ i \in \mathcal{N}.
    \end{equation}
    
    
    \item \textbf{Incentive Compatibility (IC):}  Each organization must always prefer the contract designed for its type, over all other contracts. This ensures that an organization can achieve a better utility only if it chooses the contract designed for its type.
    
    \begin{align}
    \hspace{-0.2cm}\theta_i r_i(k_i) - c_i(k_i) + T(k_i) &\ge \theta_i r_j(k_j) - c_j(k_j) + T_j(k_j), \nonumber \\
        &~~ \forall i ~, ~j \in \mathcal{N}, ~i \ne ~j.
    \end{align}
    
\end{enumerate}    

To this end, the principal's problem of finding the optimal contracts can be given as:
\begin{equation}\label{eq:optimal}
\begin{aligned}
& \underset{r_i(k_i)}{\text{max}}
&& \sum_{i=1}^{ N}\theta_i (v_i -  r_i(k_i)) \\
& \text{s.t}
&& \theta_i r_i(k_i) - c_i(k_i) + T_i(k_i) \ge 0,  \forall i \in \mathcal{N}, \\
&&& \theta_i r_i(k_i) - c_i(k_i) + T(k_i) \ge \theta_i r_j(k_j) - c_j(k_j) +\\
&&& T_j(k_j)~~,\forall i ~, ~j \in \mathcal{N}, ~i \ne ~j.   \\
\end{aligned}
\end{equation}







\section{Single Time Sharing} \label{chapter_four}

In this section, we solve the problems formulated in Section \ref{chapter_three}, of the proposed two-tier model, for the static case. In this case, we assume that the organizations share their data only once with the data collector, while the attacker tries to identify the anonymized data. In particular, we first solve the optimal contract-based model to determine the optimal rewards. Then, we derive the different Nash equilibria of the static game-theoretic model using the optimal rewards.


\subsection{Optimal Contracts}\label{sec:static_contract}

To solve the problem in \eqref{eq:optimal}, we notice that the number of constraints is large. For instance, the number of IR constraints equals $N$ and the number of IC constraints equals $N(N-1)$. However, the number of constraints can be significantly reduced as shown next.

\begin{lemma} \label{lemm_cont1}
The IC constraints are equivalent to $r_i>r_j$ for every pair of organizations such that $\theta_i > \theta_j$.
\end{lemma}

\begin{proof}
We prove this lemma by using the IC constraints for any two organizations with types $\theta_1$ and $\theta_2$ and $\theta_1 > \theta_2$. The downward and the upward IC constraints are:
\begin{align*}
&\theta_1 r_1(k_1) - c_1(k_1) + T_1(k_1) \ge \theta_1 r_2(k_2) - c_2(k_2) + T_2(k_2),\\
&\theta_2 r_2(k_2) - c_2(k_2) + T_2(k_2) \ge \theta_2 r_1(k_1) - c_1(k_1) + T_1(k_1).    
\end{align*}
By adding the two inequalities, we get:
\begin{align*}
&\theta_1 r_1(k_1) - c_1(k_1) + T_1(k_1) +  \theta_2 r_2(k_2) - c_2(k_2) + T_2(k_2) \ge \\ 
&~~\theta_1 r_2(k_2) - c_2(k_2) + T_2(k_2) + \theta_2 r_1(k_1) - c_1(k_1) + T_1(k_1),
\end{align*}
which can be rearranged as:
\begin{align*}
&\theta_1 r_1(k_1)  + \theta_2 r_2(k_2)  \ge \theta_1 r_2(k_2) + \theta_2 r_1(k_1), \\
&(\theta_1 - \theta_2) r_1(k_1) \ge (\theta_1 - \theta_2) r_2(k_2).
\end{align*}
Since $\theta_1 > \theta_2$, we can conclude that $r_1(k_1) > r_2(k_2)$.
\end{proof}

Using Lemma \ref{lemm_cont1}, we can reduce the number of IC constraints to just $N$ constraints of each two consecutive organization types. Next, we show the solution to the reduced contract problem for two organizations sharing the same data but with different anonymization values $k_L$ and $k_H$ such that $k_L < k_H$.

\begin{theorem}\label{Theorem1}
The optimal rewards for two organizations sharing the same data and using $k_L$ and $k_H$, such that $k_L < k_H$ is:
\begin{align*}
    r_H(k_H) &= \frac{c_H(k_H) - T_H(k_H)}{\theta_H}, \\
    r_L(k_L) &= \mathrm{max}\left(\frac{c_L(k_L) - T_L(k_L)}{\theta_L}, r_H(k_H)\right) \\
\end{align*}
\end{theorem}

\begin{proof}
When two organizations share the same data with different anonymization levels, the principal's problem becomes:
\begin{equation*}
\begin{aligned}
& \underset{r_i(k_i)}{\text{max}}
&& \theta_L (v -  r_L(k_L)) + \theta_H (v -  r_H(k_H)) \\
& \text{s.t}
&& \theta_L r_L(k_L) - c_L(k_L) + T_L(k_L) \ge 0, \\
&&& \theta_H r_H(k_H) - c_H(k_H) + T_H(k_H) \ge 0, \\
&&& r_L(k_L) > r_H(k_H).\\
\end{aligned}
\end{equation*}

Since $r_L(k_L)$ needs to be higher than $r_H(k_H)$ according to the last constraint, the second constraint will bind, i.e., $\theta_H r_H(k_H) - c_H(k_H) + T_H(k_H) = 0$. This will lead $r_H(k_H)$ to equal $\frac{  c_H(k_H) - T_H(k_H)}{\theta_H}$. The value of $r_L(k_L)$ needs to be at least $r_H(k_H)$ and needs to satisfy the first constraint, therefore, $r_L(k_L)$ will be the higher of both values, i.e., $\mathrm{max}\left(\frac{c_L(k_L) - T_L(k_L)}{\theta_L}, r_H(k_H)\right)$.
\end{proof}

From theorem \ref{Theorem1}, we notice that the organization with higher $k_H$ will get a reward that makes its utility equals zero, i.e., not benefiting nor losing from the sharing. However, this might be seen as being not an enough incentive for the organizations to share their data. Therefore, in the dynamic case in Section \ref{chapter_five}, we discuss long-term contract to enable the data collector to design incentivizing contracts for the organizations when they sign a long-term contract. Note in theorem \ref{Theorem1}, the solution is shown for the case of two organizations' types, however, the solution of the general case of multiple organization types will follow similarly. In that case, the organization with the highest $k$ value will get a reward that makes its utility zero. Each of the remaining organizations will get a reward that equals the maximum of its cost and the higher organizations' reward.

Finally, we note that the solution presented in theorem \ref{Theorem1} represents the case where both organizations share the same dataset. However, for the general case when each organization has a different value of $v_i$, the principal needs to solve the general problem in \eqref{eq:optimal} using any optimization technique, e.g., linear programming. In such case, each organization will achieve a different utility that does not need to equal zero.

\subsection{Proposed Game Solution}\label{sec:static_game}

After determining the rewards for each anonymization level, we can use these values in solving the formulated game. Recall that, in the organizations' utilities in \eqref{eq:updatedutility}, each organization obtains a fraction of the reward  $r_i(k_i)$ based on the attack's success probability. Let
\begin{dmath}
     \delta = (1- \frac{p\textsubscript{max}(a)}{\alpha k_i+1}) \cdot (1- \frac{p\textsubscript{max}(a)}{\alpha k_{-i}+1})
\end{dmath}
be the minimum fraction of $r_i(k_i)$ that an organization can achieve based on the maximum success probability of the available attacks, i.e., $p\textsubscript{max}(a)$. We refer to $\delta r_i(k_i)$ as the minimum profit factor.

 In the proposed game, the goal of each player is to take actions to maximize its utility given the actions of other players. When no player can improve its utility by unilaterally changing its actions, the game is said to be at equilibrium. The notion of equilibrium, in game theory, is referred to as Nash equilibrium \cite{sengupta2009game}. Nash equilibrium can be either pure Nash equilibrium, or mixed-strategy Nash equilibrium. A pure strategy equilibrium, is when every player has only one action/ strategy at equilibrium. On the other hand, a mixed Nash equilibrium represents a probability distribution over each player's set of available actions~\cite{eldosouky2016single}. 
 
 The proposed game, under the static scenario, is a finite static non-zero-sum game which is known to have a Nash equilibrium, either pure or mixed-strategy\cite{han2012game}. For the sake of analytical tractability, we consider the case in which each organization can choose between two different $k$ values, i.e., $k_L$ and $k_H$. These values represent choosing a low and high values for $k$, respectively. Based on these values, each organization will have two minimum profit factors $\delta_H$ and $\delta_L$ corresponding to the choice of $k_L$ and $k_H$, respectively.

Let $p_1$ be the probability for the first organization to choose $k_L$ such that it chooses $k_H$ with a probability $1-p_1$. Similarly, the second organization can choose $k_L$ and $k_H$ with the probabilities $p_2$ and $1-p_2$, respectively. The attacker, on the other hand, will have a probability distribution of $q_B,q_H,q_N$ for choosing the actions $B$, $H$, and $N$, respectively. We start the analysis by considering the cases in which the game $\mathcal{G}$ can have a pure strategy Nash equilibrium.

\begin{proposition}\label{prop1}
Let $k_{i}^* = \mathrm{argmax}_{k_i} r_i(k_i) - c_i(k_i)+\gamma \cdot k_i$. 
Then, the tuple $(k_i^*,k_{-i}^*,N)$ constitute a pure strategy Nash equilibrium for $\mathcal{G}$ when the attacker cannot achieve a positive utility.
\end{proposition}

\begin{proof}
We note that the attacker's utility for no-attack is zero, i.e., $u_a(k_1,k_2,N)=0$. The attacker can only turn to this choice if all the other actions yield a negative utility, i.e., all the utility instances for choosing $B$ and $H$ with the different combinations of $k_L$ and $k_H$, for each organization, will result in a negative attacker's utility. Therefore, choosing the action $N$ will be a dominant strategy for the attacker. In this case, each organization's utility will be:
\begin{equation}\label{eq:noAttack}
    u_i(k_i,k_{-i},N) = r_i(k_i) - c_i(k_i)+\gamma \cdot k_i,
\end{equation}
which clearly depends only on the organization's action and not on the other players' actions. In this case, each organization will chose the value of $k$ that maximizes its utility in \eqref{eq:noAttack}. Hence, $k_i^* = \mathrm{argmax}_{k_i} r_i(k_i) - c_i(k_i)+\gamma \cdot k_i$ will represent the optimal organization's choice under no-attack scenario. In this case, no player will have an incentive to change its choice and, therefore, the actions tuple $(k_i^*,k_{-i}^*,N)$ is a pure strategy Nash equilibrium for the game.
\end{proof}
From Proposition \ref{prop1}, the attacker's probability $q_N$ of choosing the action $N$ will be either $1$ or $0$ based on whether the action $N$ dominates the other actions or it is being dominated by another action. Thus, the actions $B$ and $H$ can be selected by probabilities $q$ and $1-q$, respectively, when $N$ is not selected. 

Similar to the attacker, each organization can have a dominant strategy under some circumstances and, hence, the probability $p$ can be either $0$ or $1$ based on the dominant strategy.

\begin{proposition}\label{prop2}
Each organization will have a dominant strategy when the rewards assigned from the data collector $r_i(k_i)$ are high enough such that the minimum profit factor is the dominant factor in the organization's utility, i.e., $\delta_H r_i(k_H)  > \gamma \cdot k_H - c_i(k_H)$ and $\delta_L r_i(k_L)  > \gamma \cdot k_L - c_i(k_L)$. The dominant strategy can then be given as the solution of:
\begin{equation*}
k_i^* = \mathrm{argmax}_{i} \delta_i {r_i(k_i)} - c_i(k_i)+ \gamma \cdot k_i, ~~ i \in \{L,H\}.
\end{equation*}

\end{proposition}


\begin{proof}
The values of $\delta_H r_i(k_H) $ and $\delta_L r_i(k_L) $ represent the minimum fractions of the reward each organization can achieve, under the attacker's maximum probability of success. 
When the values of $r_i(k_i)$ are high enough to make these minimum profit factors higher than the rest of the factors of the utility, each organization can expect that any other attacker's action will not lower its utility. Thus, the organization can determine its dominant strategy while neglecting the attacker's effect.
\end{proof}

Note in Proposition \ref{prop2}, a high reward can eliminate the attacker's effect, however, it cannot be used solely to determine the organization's action as this is affected by the other factors in the organization's utility.
 
To this end, when no player has a dominant strategy, the players will randomize over their strategies using the probability distributions of the mixed-strategy Nash equilibrium. These mixed strategies can be calculated when the players are indifferent between choosing their actions, i.e., the expected utility of choosing each action will be the same. For instance, the organizations can choose their $p$ such that the attacker's expected utility from choosing the action $B$ will equal to that of choosing the action $H$. The attacker's expected utility from choosing the action $B$ can be given by:
 \begin{align}\label{eq:expUtil1}
     \EX(u_a(k_1,k_2,B)) &= p \cdot p \cdot u_a(k_L,k_L,B) + p \cdot (1-p) \cdot \nonumber \\
     & u_a(k_L,k_H,B) + (1-p) \cdot p \cdot u_a(k_H,k_L,B)    \nonumber \\ 
     &+ (1-p) \cdot (1-p) \cdot u_a(k_H,k_H,B).
 \end{align}
 
 Similarly, the expected utility of choosing action $H$ is:
 \begin{align}\label{eq:expUtil2}
     \EX(u_a(k_1,k_2,H)&) =  p \cdot p \cdot u_a(k_L,k_L,H) + p \cdot (1-p) \cdot \nonumber \\ 
     & u_a(k_L,k_H,H) + (1-p) \cdot p \cdot u_a(k_H,k_L,H)  \nonumber \\ 
     &+ (1-p) \cdot (1-p) \cdot u_a(k_H,k_H,H).
 \end{align}

For the attacker to be indifferent between its actions, the utility in \eqref{eq:expUtil1} must equal the utility in \eqref{eq:expUtil2}. Solving both equations together, the organizations' probabilities of choosing $k_L$, i.e., $p$ can then be given as the solution of the equation:

\small
\begin{align}\label{eq:p}
\Big(&\frac{2p^2(B)-2p^2(H)}{(\alpha k_L + 1)(\alpha k_H + 1)}  - \frac{p^2(B) - p^2(H)}{(\alpha k_H + 1)^2} - \frac{p^2(B) - p^2(H)}{(\alpha k_L + 1)^2} \Big) R_a \nonumber  \\
 &  p^2 + 2 R_a p \Big( \frac{p(B)-p(H)}{(\alpha k_L + 1)}  - \frac{p(B)-p(H)}{(\alpha k_H + 1)}  + \frac{p^2(B)-p^2(H)}{(\alpha k_H + 1)^2} - \nonumber \\
 &  \frac{p^2(B)-p^2(H)}{(\alpha k_L + 1)(\alpha k_H +1)} \Big)   + \Big(\frac{2p(B)-2p(H)}{(\alpha k_H + 1)} -\frac{p^2(B)-p^2(H)}{(\alpha k_H + 1)^2} \Big) \nonumber \\ 
 &  R_a  - 4c_a(B) + 4c_a(H) =0.
\end{align}
\normalsize

After calculating the probability $p$, the attacker's probability $q$ can be calculated in a similar way by considering the expected utility of one of the organizations. Note that, due to the symmetry between the organizations, considering the utilities of both organizations will be redundant. To this end, the first organization's expected utility from choosing $k_L$ can be given by:
\begin{align}\label{eq:expUtildef1}
     \EX(u_1(k_L,k_2,a)) &= p \cdot q \cdot u_1(k_L,k_L,B) + p \cdot (1-q) \cdot   \nonumber \\
    & u_1(k_L,k_L,H) + (1-p) \cdot q \cdot u_1(k_L,k_H,B)    \nonumber \\
    &+(1-p) \cdot (1-q) \cdot u_1(k_L,k_H,H).
 \end{align}

 Similarly, the first organization's expected utility from choosing $k_H$ can be given by:
 \begin{align}\label{eq:expUtildef2}
     \EX(u_1(k_H,k_2,a)) &= p \cdot q \cdot u_1(k_H,k_L,B) + p \cdot (1-q) \cdot \nonumber \\ 
     & u_1(k_H,k_L,H) + (1-p) \cdot q \cdot u_1(k_H,k_H,B)  \nonumber \\
     & + (1-p) \cdot (1-q) \cdot u_1(k_H,k_H,H).
 \end{align}
 
For the organization to be indifferent between its actions, the utility  in  \eqref{eq:expUtildef1} must equal the utility in \eqref{eq:expUtildef2}. Solving  both equations together, the attacker’s probabilities of choosing $B$, i.e., $q$ can then be given as the solution of the equation:

\footnotesize
\begin{align}\label{eq:q}
  q &= \bigg(u_{1}(k_{L},k_{H},H)+u_{1}(k_{H},k_{H},H) - p \Big( u_{1}(k_{L},k_{L},H) + u_{1}(k_{H},\nonumber \\
 & k_{L},H)  + u_{1}(k_{L},k_{H},H) + u_{1}(k_{H},k_{H},H)\Big)\bigg) \bigg/ \bigg(u_{1}(k_{L},k_{H},H)+  \nonumber \\
 & u_{1}(k_{H},k_{H},H)-u_{1}(k_{H},k_{H},B)- u_{1}(k_{H},k_{L},B) + p \Big(u_{1}(k_{H},k_{H}, \nonumber \\
 &B) +  u_{1}(k_{H},k_{H},B)+  u_{1}(k_{L},k_{L},B)+u_{1}(k_{H},k_{L},B)-u_{1}(k_{H},k_{H}, \nonumber \\
 &H)-u_{1}(k_{L},k_{H},H)-u_{1}(k_{L},k_{L},H)-u_{1}(k_{H},k_{L},H) \Big) \bigg).
 \end{align}
 \normalsize
Given the value of $p$ from \eqref{eq:p}, the value of $q$ can be uniquely computed from \eqref{eq:q}. The Nash equilibrium mixed strategies can then be given as $(p, 1-p)$ for the organizations and $(q, 1-q)$ for the attacker. 

Note that, the solution of equations \eqref{eq:p} and \eqref{eq:q} gives one Nash equilibrium of the game $\mathcal{G}$. However, the game as involving multiple players might have other Nash equilibrium solutions \cite{lee2003solving}. These equilibrium points will have different values for $p$ and $q$ but the same outcomes for the players. Here, we only consider the solution in which all the organizations adopt the same the strategy.

\section{Repeated Sharing} \label{chapter_five}

In this section, we consider the case in which the organizations share their datasets with the data collector more than once over a period of time. At each time step, the organizations share different datasets with their choice of anonymization levels. The problems formulated in Section \ref{chapter_three} are extended by adding the notion of time. In particular, we first consider the dynamic contracts offered by the data collector, then, we solve the repeated game between the organizations and an attacker.

\subsection{Dynamic Contracts}

The problem formulated and solved in the static case in \eqref{eq:optimal} can be used by the data collector, at each time step, to determine the optimal reward value that will be given to each organization. The utility function of the data collector can be extended from \eqref{cybex} to be:
\begin{equation}\label{cybexDyn}
    U_{d,T} = \sum_{t=1}^{T} \sum_{i=1}^{N} \theta_{i,t} \left( v_{i,t} -  r_{i,t}(k_{i,t}) \right),
\end{equation}
where each parameter is the same as \eqref{cybex} but considered at a specific time step $t$ and $T$ is the total number of time steps. Solving the static problem in \eqref{eq:optimal} at every time step will ensure that the data collector maximizes its reward over all time steps.

However, considering the dynamic nature of the problem in \eqref{cybexDyn}, we propose a new technique for the data collector to improve its outcome while providing more incentive for the organizations to participate. Recall that in the static case, according to Theorem \ref{Theorem1}, one organization can get a reward that equals its cost if it uses the highest anonymization level. This will make the net outcome of this organization equals zero, which can be seen as a limiting factor that might hinder the organization from participation. In the repeated case, we propose that the data collector can sign a long-term contracts with the organizations and offer a minimum net outcome for each organizations at each time step.

To this end, the individual rationality constraint in \eqref{ineq:1} will be modified to maintain a minimum outcome value, $m_i$ at each time step such that $m_i > 0$:
\begin{equation}\label{ineq:1_mod}
        \theta_i r_i(k_i) - c_i(k_i) + T_i(k_i) \ge m_i,~~ i \in \mathcal{N}.
    \end{equation}

The range of the values of $m_i$ can be calculated using the following theorem.

\begin{theorem}\label{Theorem2}
To incentivize their participation, the data collector can offer each organization $i$ a value $m_i$, at each time step, as a minimum net outcome such that $\sum_{t=1}^T m_{i,t}$ satisfies:
\begin{equation*}
    0 < \sum_{t=1}^{T} m_{i,t} < \sum_{t=1}^{T}\theta_{i,t} \left( v_{i,t} -  r_{i,t}(k_{i,t}) \right) , t \in \mathcal{T},
\end{equation*}
where $\mathcal{T}$ is the set of the time steps that the organization was expected to have a zero utility in the static case.
\end{theorem}

\begin{proof}
Considering the problem in \eqref{eq:optimal}, the data collector will either gain $\theta_i \left( v_i -  r_i(k_i) \right)$ when an organization shares its dataset, or zero if it does not share. The idea here is to allow the data collector to offer part of its expected outcome to the organizations, to ensure their participation. In return, the data collector will increase its total reward by the difference between what it pays $m_{i,t}$ and the newly secured benefit $\theta_{i,t} \left( v_{i,t} -  r_{i,t}(k_{i,t}) \right)$ for this specific time step.

Let $\mathcal{T}$ be the set of time steps in which the individual rationality constraint will be binding in \eqref{eq:optimal}. The data collector can assume that the organizations will not be willing to participate at these time steps. The expected gain that the data collector will miss at these time steps can be expressed as follows:
\begin{equation*}
    u_{m,i} = \sum_{t=1}^{T}\theta_{i,t} \left( v_{i,t} -  r_{i,t}(k_{i,t}) \right),~ t \in \mathcal{T},
\end{equation*}
where $u_{m,i}$ is the gain that the data collector can miss if the organization $i$ does not participate.

It is clear that the data collector can improve its outcome if it can obtain part of $u_m$. This can be secured by signing a long-term contract with the organizations to offer a minimum reward that is less than $u_{m,i}$.
\end{proof}

\subsection{Dynamic Game Solution}

In this section, we consider the case in which the same game is repeated over time. This is different from the static case considered in Section~\ref{chapter_four} in that the players' utilities change over time. In particular, the organizations' utility will be:
\begin{align}\label{eq:dynamicutility}
    &u_i(k_i,k_{-i},a,t_i) = \frac{R_i(t_i)}{k_i(t_i)} \cdot (1- \frac{p(a)}{\alpha k_i(t_i)+1}) \cdot \nonumber \\
    &(1- \frac{p(a)}{\alpha k_{-i}(t_i)+1}) - c_i(t_i)(k_i(t_i))+\gamma(t_i) \cdot k_i(t_i),
\end{align}
where every parameter will be a function in the time step $t_i$. However, from a practical point of view, the probability of a successful attack $p(a)$ is expected to be constant over time, thus, it is not defined as a function in time. Note that, the functions $R_i(t_i)$ and $c_i(t_i)$ will be evaluated each time step. Thus, their values will change based on the datasets shared at a given time step. On the other hand, the coefficient of trust, $\gamma(t_i)$ will need to be a function that changes over time.

Here, we propose to define that the coefficient of trust, $\gamma(t_i)$, based on the organization's observation of the security level of the common platform.  When an organization shares its dataset with the common platform, it is the responsibility of the platform to maintain and protect the data. Thus, an organization can build its trust in the common platform based on the rate of successful and unsuccessful attacks in the previous time steps. Thus, the coefficient of trust, $\gamma(t_i)$ can be given as:
\begin{equation}\label{eq:trustdynamic}
    \gamma(t_i) = \frac{\gamma(t_{i-1}) + \mathbbold{1}_{\textrm{un}}}{2},
\end{equation}
where $\mathbbold{1}$ is the indicator function which equals $1$ in case of unsuccessful attack and equals $0$ if the attack was successful. 

Following \eqref{eq:trustdynamic}, the coefficient of trust will equal half of the value in the previous step if an attack was successful. On the other hand, if an attack was blocked by the common platform, the coefficient of trust will be the average of the previous value and $1$. This insures that the value of the coefficient remains between $0$ and $1$, when the initial value is less than $1$. Note that, the design of the trust function in \eqref{eq:trustdynamic} will punish the successful attack more than it rewards blocking the attacks. Note that, this design is different from other trust coefficients proposed in literature, e.g., \cite{khan2018using} which treat positive and negative updates equally. Our proposed coefficient ensures that the data collector does its best to protect the data and prevent any undesired access to it.

\begin{figure}[tbp]
    \centering
    \includegraphics[width=8.9cm]{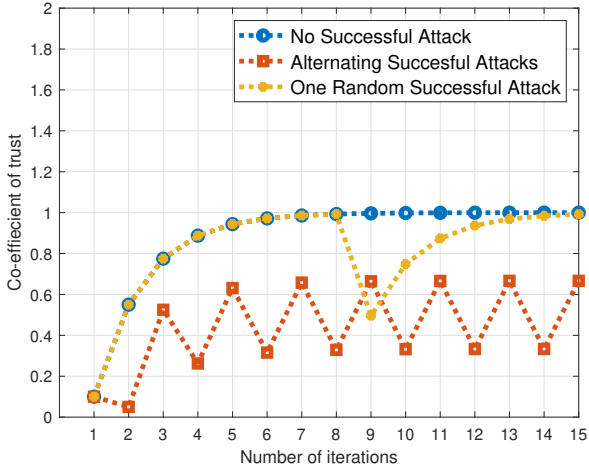}
    \caption{Coefficient of trust updates over the time}
    \label{fig:gamma_vs_time}
    \vspace{-0.4cm}
\end{figure}

In Fig.~\ref{fig:gamma_vs_time}, we study the evolution of the coefficient of trust over time. We model three scenarios, the first when there is no successful attack. In this case, the coefficient of trust will will increase from its initial value, $0.1$ until it reaches almost $1$ after $6$ times steps. This means after six times of sharing, the organization will be able to fully trust the data collector. The second scenario, is when there is a successful attack after every alternate iteration, data share. Since our coefficient punishes the successful attacks more than it rewards the positive ones, we notice that the value of the coefficient of trust will oscillate between $0.35$ and $0.65$. Finally, we consider the scenario when there is only one random successful attack, after the value has reached its maximum. We notice that this successful attack will cause the value of the coefficient of trust to drop to about $0.5$ and it will need $6$ iterations to return to $1$ again. The effect of the coefficient of trust evolution over time is studied in the simulations section.

\section{Simulation Results and Analysis}\label{chapter_six}

For our simulations, we choose two values of $k$, the lower $k_L=3$ and the higher $k_H=7$. These values are kept the same for all the experiments in this section. The values of the other parameters will be as follows, unless otherwise stated. The measure of information security, $\alpha=0.9$ and the coefficient of trust, $\gamma=1.0$. The success probabilities of the different attack types are assumed to be $P(H_d)=0.2$, $P(H_s)=0.6$, and $P(B)=0.5$, which follows the discussion in Section \ref{Sec:first_Tier} about the relation between the different probabilities. Finally, we assume similar dataset structures between the different organizations, i.e., the organizations have the same value of $V$ in \eqref{eq:cost1} so that the cost function is affected only by the choice of $k$. We note that the reward given to each organization, by the data collector, depends on the value of its shared dataset as in \eqref{cybex}. However, as these values can fall within a wide range that will affect the game equilibrium, we consider the case in which two organizations are sharing datasets with the same values. Thus, the effect of different dataset values is eliminated on the equilibrium analysis.

\begin{table*}[tbp]
\centering \caption{Attacker's equilibrium strategies}
\begin{tabular}{|c|c|c|c|c|c|c|c|c|c|c|}
\hline
$R_a$         & \textbf{10} & \textbf{20} & \textbf{30} & \textbf{40} & \textbf{50} & \textbf{60} & \textbf{70} & \textbf{80} & \textbf{90} & \textbf{100} \\ \hline
$B$ & 0           & 0           & 0.33        & 0.89        & 0.9649      & 0.8719      & 0.8153      & 0.7768      & 0.7489      & 0.7276       \\ \hline
$H$ & 0           & 0           & 0.67        & 0.11        & 0.0351      & 0.1281      & 0.1847      & 0.2232      & 0.2511      & 0.2724       \\ \hline
$N$ & 1           & 1           & 0           & 0           & 0           & 0           & 0           & 0           & 0           & 0            \\ \hline
\end{tabular} \label{Tab1}
\vspace{-0.3cm}
\end{table*}

\begin{table*}[tbp]
\centering \caption{Organization's equilibrium strategies}
\begin{tabular}{|c|c|c|c|c|c|c|c|c|c|c|}
\hline
$r(k_H)$ & \textbf{8} & \textbf{16} & \textbf{24} & \textbf{32} & \textbf{40} & \textbf{48} & \textbf{56} & \textbf{64} & \textbf{72} & \textbf{80} \\ \hline
$r(k_L)$ & \textbf{9} & \textbf{18} & \textbf{27} & \textbf{36} & \textbf{45} & \textbf{54} & \textbf{63} & \textbf{72} & \textbf{81} & \textbf{90} \\ \hline
$k_L$ & 0          & 1           & 0.362       & 0.2613      & 0.6853      & 0.7087      & 0.7241      & 0.7352      & 0.7435      & 0.7499      \\ \hline
$k_H$ & 1          & 0           & 0.638       & 0.7387      & 0.3147      & 0.2913      & 0.2759      & 0.2648      & 0.2565      & 0.2501      \\ \hline
\end{tabular}\label{Tab2}
\vspace{-0.3cm}
\end{table*}

First, we solve the formulated game $\mathcal{G}$ using the analysis in section \ref{sec:static_game}. We consider the value of organizations' datsets to vary from $10$ to $100$. These values represent the monetary rewards the data collector will give to the organizations as a reward for sharing the data. Here, we use abstract values. However, in a real-life scenario, the data collector needs to estimate these values to be proportional to the cost. Throughout the following results, we assume that the attacker can achieve the full value of the information, i.e., $R_a = v $, while the organizations will be given a reward of $r(k_L) = 0.9v$ and $r(k_H) = 0.8v$ such that $r(k_L) > r(k_H)$ according to Lemma \ref{lemm_cont1}. Note that, under a limited number of simulation parameters, there was no feasible solution for equations \eqref{eq:p} and \eqref{eq:q}. In such cases, we used a numerical solver for the game we have chosen the equilibrium point in which the organizations have the same strategy.

Using the previous parameters, the equilibrium strategies for both the attacker and the organizations are shown in Tables \ref{Tab1} and \ref{Tab2}, respectively. We note that, when the values of $R_a$ are less than $25$, the attacker cannot achieve a positive utility. Hence, it will choose not to attack. This situation corresponds to the case of Proposition {\ref{prop1}} and the organization's utility is calculated using {\eqref{eq:noAttack}}. In this case, the organization will have a pure strategy of choosing $k_H$ when $R_a=10$ and a pure strategy of choosing $k_L$ when $R_a=20$. This change occurs as $k_L$ achieves a higher reward for the organization starting from $R_a=20$, i.e., if there were no attacks for higher rewards, the organization will choose $k_L$. For the values of $R_a$ between $30$ and $100$, both the attacker and the organization will have mixed strategies, i.e., choosing their actions with certain probabilities. In the case where $R_a$ is $30$, the attacker has a higher probability of choosing homogeneity attack. Correspondingly, the organization will prioritize using $k_H$. However, for large values of $R_a$, the attacker will benefit if it performed the background knowledge attack, in this case the organization can choose between the two values of $k$ with $k_L$ being superior, i.e., it has a high probability to be chosen.

\begin{figure}[tbp]
    \centering
    \includegraphics[width=8cm]{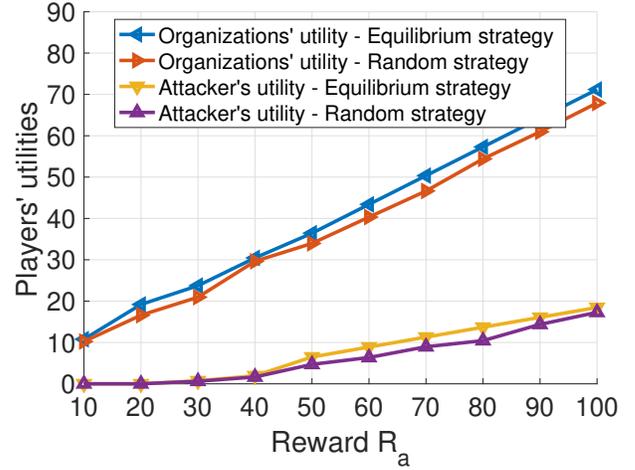}
    \caption{The organization's and the attacker's utilities at equilibrium at different reward $R$ values.}
    \label{fig:R_vs_U1_U2}
    \vspace{-.3cm}
\end{figure}

Next, we study the utilities associated with the previous equilibrium strategies in Fig. \ref{fig:R_vs_U1_U2}. The expected utilities, i.e., the summation of every outcome multiplied by the equilibrium probabilities of choosing these outcomes from Tables \ref{Tab1} and \ref{Tab2}, of the players are shown in Fig. \ref{fig:R_vs_U1_U2}. These utilities represent the outcomes of the game which each player will achieve. In Fig. \ref{fig:R_vs_U1_U2}, we can see that when the attacker chooses not to attack, its utility will equal zero. Meanwhile, the organization will be able to achieve a utility slightly higher than the reward value. On the other hand, for the reward values $R_a \ge 30$, the utility of the organization will be less than the reward as the attack reduces the organization's utility according to \eqref{eq:defenderPayOff}. However, for all the values of $R_a$, the players' utilities witness a monotonic increase in the value of $R_a$. Then, we compare these equilibrium utilities to the case where one player chooses random probabilities while the other player sticks to its equilibrium strategy. From Fig. \ref{fig:R_vs_U1_U2}, we can see that when a player deviates from the equilibrium strategy, to a random strategy, it cannot achieve a higher utility as its utility will be lower or equal to the equilibrium utility. This corroborates the importance of finding Nash equilibrium strategies as they represent the best that each player can achieve given their opponent's actions.

In Fig. \ref{fig:B_vs_p_q}, we study the effect of the success probability of the background knowledge attack, i.e., $p(B)$ on the equilibrium strategies of the players, under two different scenarios of lower reward ($R_a = 50$) and a higher reward ($R_a = 100$). All the other parameters are the same as in Fig. \ref{fig:R_vs_U1_U2}. Note that, the values of $p(B)$ are chosen to start at $0.3$ to satisfy the assumption $p(B)>p(H_d)$. For each value of $p(B)$ and $R_a$, the game $\mathcal{G}$ is solved and the equilibrium strategies are shown in Fig. \ref{fig:B_vs_p_q} in a similar way to the values in Tables \ref{Tab1} and \ref{Tab2}. From Fig. \ref{fig:B_vs_p_q}, we can see that when $p(B)$ is slightly higher than $p(H_d)$ i.e., $p(B)=0.3$ the attacker will have a zero value of $q$ which corresponds to exclusively choosing to perform homogeneity attack, for both $R_a = 50$ and $R_a = 100$. At the same point, the organization will choose $k_L$ with slightly higher probability for both $R_a = 50$ and $R_a = 100$. However, as the value of $p(B)$ increases, under lower reward, the organization will prefer to use $k_H$ as it provides a higher utility. Correspondingly, the attacker's probability for performing background knowledge attacks increases till it becomes close to pure strategy. On the other hand, under high reward ($R_a=100$), the attacker's probability of performing background knowledge attack will become near pure-strategy due to the increased probability of success. However, the organization's probability of choosing $k_L$ will be very high, i.e, near pure-strategy. In Fig. \ref{fig:B_vs_p_q}, the effect of the reward that each organization is given is clear on its equilibrium strategies, as for low rewards, the organization will be choosing $k_H$ while for higher rewards the organization will be choosing $k_L$.



\begin{figure}[tbp]
    \centering
    \includegraphics[width=8.5cm]{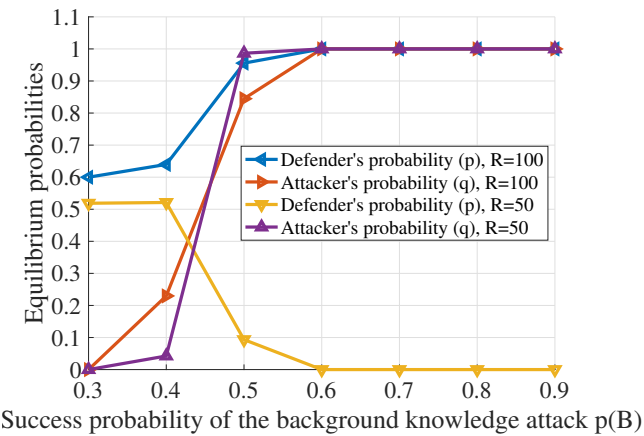}
    \caption{The organization's and the attacker's equilibrium probabilities under low and high rewards at different success probabilities for background knowledge attack $p(B)$ values.}
    \label{fig:B_vs_p_q}
    \vspace{-.3cm}
\end{figure}

In Fig. \ref{fig:H1_vs_p_q}, we study the effect of the success probability of the homogeneity attack, at similar values of $k$, i.e, $p(H_s)$ on the equilibrium strategies of the players, under the same varying reward scenarios as the Fig \ref{fig:B_vs_p_q}. Similar to Fig. \ref{fig:B_vs_p_q}, the values of $p(H_s)$ are starting at $0.3$ so that $p(H_s)>p(H_d)$. The rest of the simulation parameters are the same as Fig. \ref{fig:R_vs_U1_U2}. From Fig. \ref{fig:H1_vs_p_q}, we can see that when $p(H_s)$ is less than $p_(B)$ i.e., $p(H_s) <0.5$, under low reward, the attacker will have a near pure-strategy probability of choosing the background knowledge attack. This probability will decrease as $p(H_s)$ is equal to $p(B)$ or higher. In this case, the attacker will choose the homogeneity attack with higher probability especially with the increase in its success probability. Similarly, when $p(H_s) <0.5$, the organization will choose $k_H$ with higher probability. However, this probability will be decreasing as $p(H_s)$ increases. Under the high reward scenario, when $p(H_s) <0.5$, the attacker will be performing background knowledge attack in a near exclusive manner, while as $p(H_s)$ is equal to $p(B)$ or higher, this probability decreases. Similarly, when $p(H_s) <0.5$, the organization will have a higher probability on $k_L$. As $p(H_s)$ increases, the organizations will place more emphasis on playing $k_H$. We note that, the low values of $p(H_S) \le 0.5$ are similar to the case of Proposition \ref{prop2}, where the organizations are able to play a pure strategy that maximize their utility. Consequently, the attacker also plays a pure strategy of choosing the background knowledge attack. However, as the values of $p(H_s)$ increase, the organizations' portion of the utility decrease and the game witnesses a mixed strategy solution.


\begin{figure}[tbp]
    \centering
    \includegraphics[width=8cm]{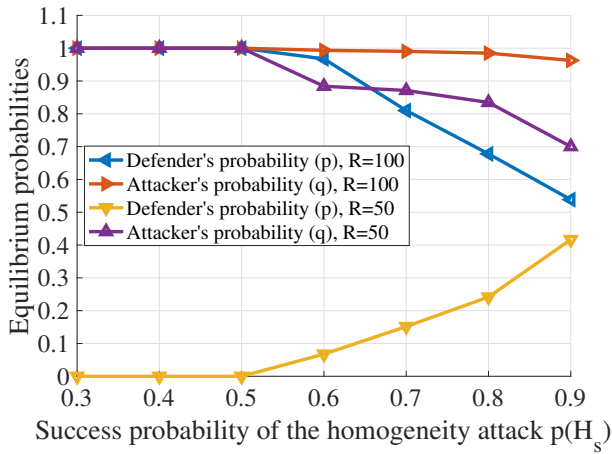}
    \caption{The organization's and the attacker's equilibrium probabilities under low and high rewards at different success probabilities for homogeneity attack $p(H_s)$ values.}
    \label{fig:H1_vs_p_q}
    \vspace{-.3cm}
\end{figure}

\begin{figure}[tbp]
    \centering
    \includegraphics[width=8cm]{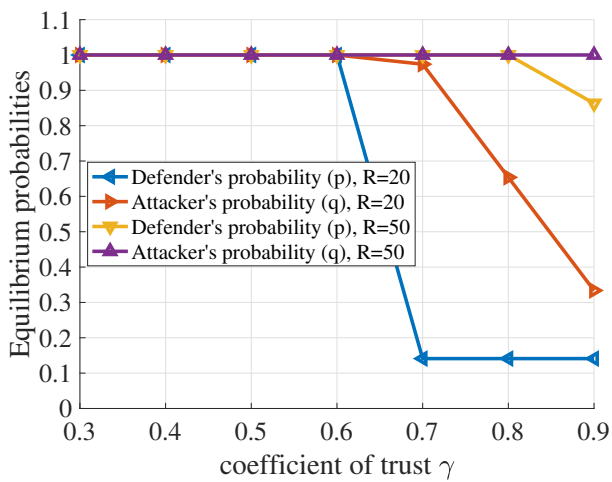}
    \caption{The organization's and the attacker's equilibrium probabilities under low and high rewards at different coefficient of trust $\gamma$ values.}
    \label{fig:gamma_vs_p_q}
    \vspace{-.3cm}
\end{figure}

Next, we study the dynamic case of the game in Fig. \ref{fig:gamma_vs_p_q}. We notice from \eqref{eq:trustdynamic} that the coefficient of trust changes over time based on its previous values that depend on the breaches in the previous steps. However, these successful breaches cannot be predicted beforehand and they can follow different behavior as in Fig. \ref{fig:gamma_vs_time}. Therefore, in Fig. \ref{fig:gamma_vs_p_q}, we show the equilibrium strategies against different $\gamma$ values. In a dynamic scenario, the equilibrium at each time step can be identified through the corresponding $\gamma$ in Fig. \ref{fig:gamma_vs_p_q}. Similar to the previous figures, we study the effect of $\gamma$ under low and high values of $R_a$, while the rest of the parameters is similar to Fig. \ref{fig:R_vs_U1_U2}. In Fig. \ref{fig:gamma_vs_p_q}, We can see that, for higher rewards $R_a = 50$, the trust factor $\gamma$ does not have a visible effect on the equilibrium strategies, i.e, both players' will choose a pure strategy (except for the organization at $\gamma = 0.9$). However, for lower rewards $R_a = 20$, the equilibrium strategies change. For instance, the equilibrium probabilities of both the attacker and the organization deviate significantly as $\gamma > 0.6$. The difference between the high and the low reward values is that in the lower reward case, the trust factor represents a significant portion of the utility, and, thus, it affects the equilibrium strategy. However, for high values of $R_a \ge 50$, the coefficient of trust can be seen negligible compared to the reward value, and, thus, it does not affect the equilibrium.




Finally, in Fig. \ref{fig:contract1}, we study the contract utilities of different organizations when they accept contracts from the data collector that match their anonymization levels and that do not match. For this case, we consider three organizations adopting three anonymization levels, $k_L, k_M, k_H$. We use a generalized solution of Theorem \ref{Theorem1} to assign the contract utilities of the organizations, while maintaining the incentive compatibility constraints between the different organization's types.
In Fig. \ref{fig:contract1}, we can clearly see that it is better for each organization to choose the contract designed for its type as it will benefit with the most utility if they stick to their incentive compatible contract, versus if they choose to accept a reward designed for another anonymization level, whether higher or lower, as it makes their net outcome (utility) lower.

\begin{figure}[tbp]
    \centering
    \includegraphics[width=8cm]{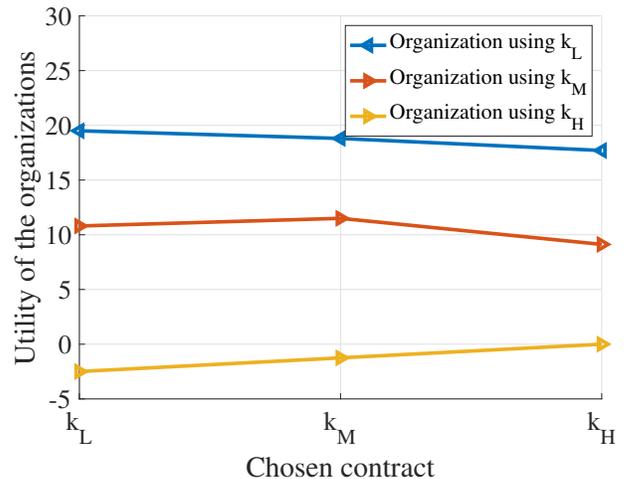}
    \caption{The utility of each organization while accepting the contract designed for its type or other contracts.}
    \label{fig:contract1}
    \vspace{-.3cm}
\end{figure}

\section{Conclusions} \label{chapter_seven}

In this paper, we have proposed a two-tier model to study the interactions between different organizations who share anonymized datasets with a data collector, and an attacker who performs de-anonymization attacks. In the first tier, a game-theoretic model has been used to determine the optimal anonymization level for k-anonymization, under possible attacks. Two common attack types have been considered which are homogeneity attack and background knowledge attack. In the second tier, a contract-theoretic model has been formulated to determine the data collector's optimal rewards that are given to the organizations for their data. In both levels of the problem, the mathematical solutions have been derived and closed-form solutions have been derived. The problem has also been studied in both a static scenario and a dynamic scenario to highlight the effect of repeated sharing on the players' behavior. Through simulation results, we have shown that the sharing organizations can optimize their anonymization level selection while the data collector can economically benefit by incentivizing more organizations to share their data.

\bibliography{mybibfile}

\begin{thebibliography}{10}

\bibitem{korta2020Every}
A.~{Kotra}, A.~{Eldosouky}, and S.~{Sengupta}, ``Every anonymization begins
  with k: A game-theoretic approach for optimized k selection in
  k-anonymization,'' in {\em 2020 International Conference on Advances in
  Computing and Communication Engineering (ICACCE)}, pp.~1--6, 2020.

\bibitem{zyskind2015decentralizing}
G.~Zyskind, O.~Nathan, {\em et~al.}, ``Decentralizing privacy: Using blockchain
  to protect personal data,'' in {\em 2015 IEEE Security and Privacy
  Workshops}, pp.~180--184, IEEE, 2015.

\bibitem{badsha2019privacy}
S.~Badsha, I.~Vakilinia, and S.~Sengupta, ``Privacy preserving cyber threat
  information sharing and learning for cyber defense,'' in {\em 2019 IEEE 9th
  Annual Computing and Communication Workshop and Conference (CCWC)},
  pp.~0708--0714, IEEE, 2019.

\bibitem{soria2017individual}
J.~Soria-Comas, J.~Domingo-Ferrer, D.~S{\'a}nchez, and D.~Meg{\'\i}as,
  ``Individual differential privacy: A utility-preserving formulation of
  differential privacy guarantees,'' {\em IEEE Transactions on Information
  Forensics and Security}, vol.~12, no.~6, pp.~1418--1429, 2017.

\bibitem{zhu2014correlated}
T.~Zhu, P.~Xiong, G.~Li, and W.~Zhou, ``Correlated differential privacy: Hiding
  information in non-iid data set,'' {\em IEEE Transactions on Information
  Forensics and Security}, vol.~10, no.~2, pp.~229--242, 2014.

\bibitem{keshavarz2020real}
M.~Keshavarz, A.~Shamsoshoara, F.~Afghah, and J.~Ashdown, ``A real-time
  framework for trust monitoring in a network of unmanned aerial vehicles,'' in
  {\em IEEE INFOCOM 2020-IEEE Conference on Computer Communications Workshops
  (INFOCOM WKSHPS)}, pp.~677--682, IEEE, 2020.

\bibitem{afghah2020cooperative}
F.~Afghah, A.~Shamsoshoara, L.~L. Njilla, and C.~A. Kamhoua, ``Cooperative
  spectrum sharing and trust management in iot networks,'' {\em Modeling and
  Design of Secure Internet of Things}, pp.~79--109, 2020.

\bibitem{boreale2019relative}
M.~Boreale, F.~Corradi, and C.~Viscardi, ``Relative privacy threats and
  learning from anonymized data,'' {\em IEEE Transactions on Information
  Forensics and Security}, vol.~15, pp.~1379--1393, 2020.

\bibitem{Domingo2019Steered}
J.~{Domingo-Ferrer}, J.~{Soria-Comas}, and R.~{Mulero-Vellido}, ``Steered
  microaggregation as a unified primitive to anonymize data sets and data
  streams,'' {\em IEEE Transactions on Information Forensics and Security},
  vol.~14, no.~12, pp.~3298--3311, 2019.

\bibitem{eldosouky2018cybersecurity}
A.~Eldosouky and W.~Saad, ``On the cybersecurity of m-health iot systems with
  led bitslice implementation,'' in {\em 2018 IEEE International Conference on
  Consumer Electronics (ICCE)}, pp.~1--6, IEEE, 2018.

\bibitem{sweeney2002k}
L.~Sweeney, ``k-anonymity: A model for protecting privacy,'' {\em International
  Journal of Uncertainty, Fuzziness and Knowledge-Based Systems}, vol.~10,
  no.~05, pp.~557--570, 2002.

\bibitem{sweeney2002achieving}
L.~SWEENEY, ``Achieving k-anonymity privacy protection using generalization and
  suppression,'' {\em International Journal of Uncertainty, Fuzziness and
  Knowledge-Based Systems}, vol.~10, no.~05, pp.~571--588, 2002.

\bibitem{machanavajjhala2006diversity}
A.~Machanavajjhala, J.~Gehrke, D.~Kifer, and M.~Venkitasubramaniam,
  ``l-diversity: Privacy beyond k-anonymity,'' in {\em 22nd International
  Conference on Data Engineering (ICDE'06)}, pp.~24--24, IEEE, 2006.

\bibitem{li2007t}
N.~Li, T.~Li, and S.~Venkatasubramanian, ``t-closeness: Privacy beyond
  k-anonymity and l-diversity,'' in {\em 2007 IEEE 23rd International
  Conference on Data Engineering}, pp.~106--115, IEEE, 2007.

\bibitem{li2009modeling}
T.~Li, N.~Li, and J.~Zhang, ``Modeling and integrating background knowledge in
  data anonymization,'' in {\em 2009 IEEE 25th International Conference on Data
  Engineering}, pp.~6--17, IEEE, 2009.

\bibitem{wang2011enhanced}
Q.~Wang, Z.~Xu, and S.~Qu, ``An enhanced k-anonymity model against homogeneity
  attack.,'' {\em JOURNAL OF SOFTWARE}, vol.~6, no.~10, pp.~1945--1952, 2011.

\bibitem{liang2008infoloss}
Z.~Liang and R.~Wei, ``Efficient k-anonymization for privacy preservation,'' in
  {\em 2008 12th International Conference on Computer Supported Cooperative
  Work in Design}, pp.~737--742, IEEE, 2008.

\bibitem{xu2016dynamic}
L.~Xu, C.~Jiang, Y.~Qian, Y.~Zhao, J.~Li, and Y.~Ren, ``Dynamic privacy
  pricing: A multi-armed bandit approach with time-variant rewards,'' {\em IEEE
  Transactions on Information Forensics and Security}, vol.~12, no.~2,
  pp.~271--285, 2017.

\bibitem{baza2019b}
M.~Baza, N.~Lasla, M.~Mahmoud, G.~Srivastava, and M.~Abdallah, ``B-ride: Ride
  sharing with privacy-preservation, trust and fair payment atop public
  blockchain,'' {\em IEEE Transactions on Network Science and Engineering},
  2019.

\bibitem{de2017pracis}
J.~M. de~Fuentes, L.~Gonz{\'a}lez-Manzano, J.~Tapiador, and P.~Peris-Lopez,
  ``Pracis: Privacy-preserving and aggregatable cybersecurity information
  sharing,'' {\em computers \& security}, vol.~69, pp.~127--141, 2017.

\bibitem{baza2020sharing}
M.~Baza, A.~Salazar, M.~Mahmoud, M.~Abdallah, and K.~Akkaya, ``On sharing
  models instead of data using mimic learning for smart health applications,''
  in {\em 2020 IEEE International Conference on Informatics, IoT, and Enabling
  Technologies (ICIoT)}, pp.~231--236, IEEE, 2020.

\bibitem{zhao2012collaborative}
W.~Zhao and G.~White, ``A collaborative information sharing framework for
  community cyber security,'' in {\em 2012 IEEE Conference on Technologies for
  Homeland Security (HST)}, pp.~457--462, IEEE, 2012.

\bibitem{alawneh2008preventing}
M.~Alawneh and I.~M. Abbadi, ``Preventing information leakage between
  collaborating organisations,'' in {\em Proceedings of the 10th international
  Conference on Electronic Commerce}, pp.~1--10, 2008.

\bibitem{han2012game}
Z.~Han, D.~Niyato, W.~Saad, T.~Ba{\c{s}}ar, and A.~Hj{\o}rungnes, {\em Game
  theory in wireless and communication networks: theory, models, and
  applications}.
\newblock Cambridge university press, 2012.

\bibitem{das2020think}
T.~Das, A.~Eldosouky, and S.~Sengupta, ``Think smart, play dumb: Analyzing
  deception in hardware trojan detection using game theory,'' in {\em 2020
  International Conference on Cyber Security and Protection of Digital Services
  (Cyber Security)}, pp.~1--8, IEEE, 2020.

\bibitem{Eldosouky2020drones}
A.~{Eldosouky}, A.~{Ferdowsi}, and W.~{Saad}, ``Drones in distress: A
  game-theoretic countermeasure for protecting uavs against gps spoofing,''
  {\em IEEE Internet of Things Journal}, vol.~7, no.~4, pp.~2840--2854, 2020.

\bibitem{Ferdowsi2020Interdependence}
A.~{Ferdowsi}, A.~{Eldosouky}, and W.~{Saad}, ``Interdependence-aware
  game-theoretic framework for secure intelligent transportation systems,''
  {\em IEEE Internet of Things Journal}, pp.~1--1, 2020.

\bibitem{vakilinia20173}
I.~Vakilinia, D.~K. Tosh, and S.~Sengupta, ``3-way game model for
  privacy-preserving cybersecurity information exchange framework,'' in {\em
  MILCOM 2017-2017 IEEE Military Communications Conference (MILCOM)},
  pp.~829--834, IEEE, 2017.

\bibitem{wan2017expanding}
Z.~Wan, Y.~Vorobeychik, W.~Xia, E.~W. Clayton, M.~Kantarcioglu, and B.~Malin,
  ``Expanding access to large-scale genomic data while promoting privacy: a
  game theoretic approach,'' {\em The American Journal of Human Genetics},
  vol.~100, no.~2, pp.~316--322, 2017.

\bibitem{ezhei2017information}
M.~Ezhei and B.~T. Ladani, ``Information sharing vs. privacy: A game theoretic
  analysis,'' {\em Expert Systems with Applications}, vol.~88, pp.~327--337,
  2017.

\bibitem{liu2013game}
X.~Liu, K.~Liu, L.~Guo, X.~Li, and Y.~Fang, ``A game-theoretic approach for
  achieving k-anonymity in location based services,'' in {\em 2013 Proceedings
  IEEE INFOCOM}, pp.~2985--2993, IEEE, 2013.

\bibitem{chakravarthy2012coalitional}
S.~L. Chakravarthy, V.~V. Kumari, and C.~Sarojini, ``A coalitional game
  theoretic mechanism for privacy preserving publishing based on k-anonymity,''
  {\em Procedia Technology}, vol.~6, pp.~889--896, 2012.

\bibitem{Adl2012privacy}
R.~Karimi~Adl, M.~Askari, K.~Barker, and R.~Safavi-Naini, ``Privacy consensus
  in anonymization systems via game theory,'' in {\em Data and Applications
  Security and Privacy XXVI} (N.~Cuppens-Boulahia, F.~Cuppens, and
  J.~Garcia-Alfaro, eds.), pp.~74--89, 2012.

\bibitem{bolton2005contract}
P.~Bolton, M.~Dewatripont, {\em et~al.}, {\em Contract theory}.
\newblock MIT press, 2005.

\bibitem{eldosouky2020resilient}
A.~Eldosouky, W.~Saad, and N.~Mandayam, ``Resilient critical infrastructure:
  Bayesian network analysis and contract-based optimization,'' {\em Reliability
  Engineering \& System Safety}, p.~107243, 2020.

\bibitem{duan2012cooperative}
L.~Duan, L.~Gao, and J.~Huang, ``Cooperative spectrum sharing: A contract-based
  approach,'' {\em IEEE Transactions on Mobile Computing}, vol.~13, no.~1,
  pp.~174--187, 2012.

\bibitem{contract_Abdel}
A.~{Eldosouky}, W.~{Saad}, C.~{Kamhoua}, and K.~{Kwiat}, ``Contract-theoretic
  resource allocation for critical infrastructure protection,'' in {\em 2015
  IEEE Global Communications Conference (GLOBECOM)}, pp.~1--6, Dec 2015.

\bibitem{wang2017modeling}
E.~K. Wang, B.~Jia, and N.~Ke, ``Modeling background knowledge for privacy
  preserving medical data publishing,'' in {\em 2017 International Conference
  on Computer Systems, Electronics and Control (ICCSEC)}, pp.~136--141, IEEE,
  2017.

\bibitem{meyerson2004complexity}
A.~Meyerson and R.~Williams, ``On the complexity of optimal k-anonymity,'' in
  {\em Proceedings of the twenty-third ACM SIGMOD-SIGACT-SIGART symposium on
  Principles of database systems}, pp.~223--228, 2004.

\bibitem{Gordon:2002:EIS:581271.581274}
L.~A. Gordon and M.~P. Loeb, ``The economics of information security
  investment,'' {\em ACM Trans. Inf. Syst. Secur.}, vol.~5, pp.~438--457, Nov.
  2002.

\bibitem{sattar2014probabilistic}
A.~S. Sattar, J.~Li, J.~Liu, R.~Heatherly, and B.~Malin, ``A probabilistic
  approach to mitigate composition attacks on privacy in non-coordinated
  environments,'' {\em Knowledge-based systems}, vol.~67, pp.~361--372, 2014.

\bibitem{ogut2005cyber}
H.~Ogut, N.~Menon, and S.~Raghunathan, ``Cyber insurance and it security
  investment: Impact of interdependence risk.,'' in {\em WEIS}, 2005.

\bibitem{sengupta2009game}
S.~Sengupta, M.~Chatterjee, and K.~Kwiat, ``A game theoretic framework for
  power control in wireless sensor networks,'' {\em IEEE Transactions on
  Computers}, vol.~59, no.~2, pp.~231--242, 2009.

\bibitem{eldosouky2016single}
A.~Eldosouky, W.~Saad, and D.~Niyato, ``Single controller stochastic games for
  optimized moving target defense,'' in {\em 2016 IEEE International Conference
  on Communications (ICC)}, pp.~1--6, IEEE, 2016.

\bibitem{lee2003solving}
K.-H. Lee and R.~Baldick, ``Solving three-player games by the matrix approach
  with application to an electric power market,'' {\em IEEE Transactions on
  Power Systems}, vol.~18, no.~4, pp.~1573--1580, 2003.

\bibitem{khan2018using}
Z.~A. Khan, ``Using energy-efficient trust management to protect iot networks
  for smart cities,'' {\em Sustainable cities and society}, vol.~40, pp.~1--15,
  2018.

\end{thebibliography}
\bibliographystyle{ieeetr}
\end{document}